\documentclass[10pt]{article}

\usepackage[utf8]{inputenc}
\usepackage{amsmath}
\usepackage{amsthm}
\usepackage{amsfonts}
\usepackage{amssymb}
\usepackage{fullpage}
\usepackage{mathtools}
\usepackage{thm-restate}
\usepackage{mleftright}
\usepackage{authblk}
\usepackage{bm}
\usepackage[linesnumbered,ruled,lined,noend]{algorithm2e}
\usepackage[colorlinks=true, linkcolor=blue, citecolor=blue,urlcolor=black]{hyperref}
\usepackage[capitalize,noabbrev]{cleveref}
\usepackage[table]{xcolor}
\usepackage{etoolbox}
\usepackage{xparse}
\usepackage{nicefrac}
\usepackage{booktabs}
\usepackage{caption}
\usepackage{enumitem}
\usepackage{float}
\usepackage{array}
\usepackage{pifont}
\usepackage{mdframed}

\newcounter{qst}
\crefname{qst}{Question}{Questions}

\newtheorem{definition}{Definition}
\newtheorem{lemma}{Lemma}

\newtheorem{theorem}{Theorem}
\newtheorem{corollary}{Corollary}
\newtheorem{assumption}{Assumption}

\newcommand{\vc}{\vec{c}}

\DeclareMathOperator{\reg}{Reg}

\usepackage{natbib}

\title{Clairvoyant Regret Minimization: Equivalence with Nemirovski's Conceptual Prox Method and Extension to General Convex Games}
\author[1]{Gabriele Farina}
\author[2]{Christian Kroer}
\author[3]{Chung-Wei Lee}
\author[3]{Haipeng Luo}
\affil[1]{Meta AI}
\affil[2]{IEOR Department, Columbia University}
\affil[3]{Computer Science Department, University of Southern California}
\affil[ ]{\texttt{gfarina@fb.com, christian.kroer@columbia.edu, \{leechung, haipengl\}@usc.edu}}
\date{\today}

\makeatletter
    \patchcmd\algocf@Vline{\vrule}{\vrule \kern-0.4pt}{}{}
    \patchcmd\algocf@Vsline{\vrule}{\vrule \kern-0.4pt}{}{}
\makeatother
\SetKwComment{Hline}{}{\vspace{-3mm}\textcolor{gray}{\hrule}\vspace{1mm}}
\definecolor{darkgrey}{gray}{0.3}
\definecolor{commentcolor}{gray}{0.5}
\SetKwComment{Comment}{\color{commentcolor}[$\triangleright$\ }{}
\SetCommentSty{}
\SetNlSty{}{\color{darkgrey}}{}
\SetKwProg{Fn}{function}{}{}
\SetKwProg{Subr}{subroutine}{}{}
\crefalias{AlgoLine}{line}%
\crefname{algocf}{Algorithm}{Algorithms}

\newcommand*{\N}{{\mathbb{N}}}
\newcommand*{\R}{{\mathbb{R}}}

\newcommand{\lip}{L}

\newcommand{\cZ}{\mathcal{Z}}

\newcommand{\cX}{\mathcal{X}}
\newcommand{\defeq}{\coloneqq}
\renewcommand{\^}[1]{^{(#1)}}
\newcommand{\bbR}{\mathbb{R}}

\DeclareMathOperator*{\argmin}{arg\,min}
\renewcommand{\vec}[1]{\bm{#1}}

\newcommand{\vz}{\vec{z}}
\renewcommand{\reg}{\varphi}

\newcommand{\range}[1]{[\![#1]\!]}

\newcommand{\vx}{\vec{x}}
\newcommand{\vg}{\vec{g}}
\newcommand{\vw}{\vec{w}}
\newcommand{\vu}{\vec{\hat z}}
\newcommand{\prox}[2]{\Pi_{#1}\mleft(#2\mright)}

\newtheorem{observation}{Observation}

\RenewDocumentCommand{\div}{omm}{\IfNoValueTF{#1}{D(#2\,\|\,#3)}{D_{#1}(#2\,\|\,#3)}}
\renewcommand{\dim}{d}

\def\[#1\]{%
  \begin{align*}%
    #1%
  \end{align*}%
}
\NewDocumentCommand{\numberthis}{om}{%
\IfNoValueTF{#1}{%
    \refstepcounter{equation}\tag{\theequation}%
  }{%
    \tag{#1}%
  }%
  \label{#2}%
}

\begin{document}

\maketitle

\begin{abstract}
    A recent paper by \citet{Piliouras21:Optimal,Piliouras22:Beyond} introduces an uncoupled learning algorithm for normal-form games---called Clairvoyant MWU (CMWU). In this note we show that CMWU is equivalent to the \textit{conceptual prox method} described by \citet{Nemirovski04:Prox}. This connection immediately shows that it is possible to extend the CMWU algorithm to any convex game, a question left open by Piliouras et al. We call the resulting algorithm---again equivalent to the conceptual prox method---\textit{Clairvoyant OMD}. At the same time, we show that our analysis yields an improved regret bound compared to the original bound by Piliouras et al., in that the regret of CMWU scales only with the square root of the number of players, rather than the number of players themselves.
\end{abstract}

\section{Introduction}

A recent line of work has focused on identifying learning algorithms such that, when used by all players in a game, each player's regret grows polylogarithmically in the number of repetitions $T$, improving over the traditional (and unimprovable) $O(\sqrt{T})$ bounds\footnote{For simplicity, in the introduction our $O(\cdot)$ notation hides all parameters independent of $T$.} of no-regret algorithms for the more adversarial setting in which no assumption about the algorithm used by other agents is made (see, \emph{e.g.}, \citet{Daskalakis21:Fast,Farina22:Near,Daskalakis21:Near,Anagnostides21:Near,Farina22:Kernelized,Chen20:Hedging,Foster16:Learning,Syrgkanis15:Fast,Rakhlin13:Optimization,Daskalakis11:Near}). By leveraging well-known connections between regret and equilibria in games (\emph{e.g.}, \citet{Freund99:Adaptive, Hart00:Simple,Foster97:Calibrated,Roughgarden15:Intrinsic}), such learning algorithms can then be used as computational approaches to equilibrium finding, leading to $\tilde{O}(1/T)$ convergence to coarse correlated equilibria. 
This reduction from no-regret learning to equilibrium computation is largely the preferred approach in theory and practice~\citep{Bowling15:Heads,Brown17:Superhuman,Anagnostides22:Faster,Farina21:Better}.

In a recent work, \citet{Piliouras21:Optimal} depart from this no-regret learning perspective, by introducing a new algorithm that they call \emph{Clairvoyant MWU} (CMWU) for normal-form games. 
CMWU is a variant of the popular multiplicative weights updates (MWU) algorithm, where action probabilities are scaled exponentially according to their  payoff at each time step. In CMWU this scaling is assumed to be done at time $t$ with respect to a tight approximation of the payoff that the player will see \emph{at that same time $t$} (hence the adjective \emph{clairvoyant}). 
In CMWU, the iterates generated by the algorithm are not known to yield regret bounds, and as a decentralized protocol the CMWU dynamics require the players to coordinate on repeatedly computing a fixed point over a sequence of iterations. However, the CMWU dynamics, whether centralized or decentralized, are attractive as a method for computing a coarse correlated equilibrium (CCE). In particular, they yield a very competitive $O(\log T / T)$ rate of convergence, while requiring only a single gradient computation as well as a linear-time closed-form strategy update at every iteration.

\begin{table}[t]\centering
    \def\clap#1{\hbox to 0pt{\hss#1\hss}}
    \newcommand{\yes}{\ding{51}}%
    \newcommand{\no}{\textcolor{black!100}{\ding{55}}}%
    \scalebox{.945}{\begin{tabular}{m{4.3cm}m{5.3cm}>{\centering\arraybackslash}m{2.2cm}>{\centering\arraybackslash}m{2.0cm}>{\centering\arraybackslash}m{1.5cm}}
        \toprule
        \textbf{Algorithm} & \textbf{Operations required to\newline compute an $\epsilon$-CCE in NFGs} & \textbf{Are iterates known to be no-regret?} & \textbf{Generalizes to convex games?} & \textbf{Leads to CE in NFGs?}
        \\
        \midrule
        Optimistic MWU\newline\citep{Daskalakis21:Near} & $\displaystyle O\mleft(n\,d\log d\cdot \frac{1}{\epsilon} \log^4 \frac{1}{\epsilon}\mright)$ & \yes & \no & \no\\[4mm]
        BM-OFTRL-LogBar \newline\citep{Anagnostides22:Uncoupled} & $\displaystyle O\mleft(n\, \mathrm{poly}(d)\cdot \frac{1}{\epsilon}\log\mleft(\frac{1}{\epsilon}\mright)\log\log\frac{1}{\epsilon} \mright)$ & \yes & \no & \yes\\[4mm]
        LRL-OFTRL\newline\citep{Farina22:Near} & $\displaystyle O\mleft(n\, \mathrm{poly}(d)\cdot \frac{1}{\epsilon}\log\mleft(\frac{1}{\epsilon}\mright)\log\log\frac{1}{\epsilon} \mright)$ & \yes & \yes & \no\\[4mm]
        Clairvoyant MWU\newline\citep{Piliouras22:Beyond} & $\displaystyle O\mleft(n \,d \log(d)\cdot \frac{1}{\epsilon} \log \frac{1}{\epsilon}\mright)$ & ~~~~~\no\newline\small\clap{(but a subset is)} & \no & \no\\[4mm]
        Clairvoyant OMD\newline[\textcolor{purple}{this note}] & $\displaystyle O\mleft(\sqrt{n} \,d \log(d)\cdot \frac{1}{\epsilon} \log \frac{1}{\epsilon}\mright)$ & ~~~~~\no\newline\small\clap{(but a subset is)} & \yes & \no\\[4mm]
        \bottomrule
    \end{tabular}}
    \caption{Comparison of existing learning-based methods for computing an $\epsilon$-approximate coarse-correlated equilibrium (CCE) in a generic normal-form game (NFG) with $n$ players, $d$ total actions, and utilities in $[0,1]$, for values of $\epsilon = O(1/d)$. We use the notation $\mathrm{poly}(d)$ to hide polynomial terms in $d$ of order at least $2$. We remark that BM-OFTRL-LogBar is an algorithm designed to minimize the more challenging notion of \emph{swap regret}, and guarantees convergence to the set of correlated equilibrium (CE) in normal-form games.}
    \label{tab:comparison}
\end{table}

In this note, we show that the CMWU algorithm can be viewed as an instantiation of the \emph{conceptual prox method} (CPM), which has been studied extensively in the first-order methods literature~\citet{Kiwiel97:Proximal,Chen93:Convergence,Nemirovski04:Prox}. \citet{Nemirovski04:Prox} discusses the CPM as a conceptual algorithm that achieves an $O(1/T)$ rate of convergence to a solution to a monotone variational inequality (VI). Monotone VI generalizes for example the problem of computing a two-player zero-sum Nash equilibrium.
He labels it a conceptual method because each step of the algorithm requires solving a non-trivial variational inequality, which means that it is not clear that the algorithm is implementable. He then shows, however, that an approximate solution to this VI can be computed in logarithmic time in the required precision, because the VI corresponds to a solution to a fixed point of a mapping which is a contraction. Finally, he goes on to show that in fact one can get the same rate by performing only two steps of the contraction. This results in the famous \emph{mirror prox} algorithm.
What we show is that the CMWU algorithm can be viewed as a specialization of CPM to the setting where the feasible set is the Cartesian product of the player's strategy spaces, and the operator used in the variational inequality is the gradient operator for each player.
However, here one must depart from the CPM perspective of Nemirovski, because his paper focuses on \emph{monotone} operators. This is because his method ultimately requires computing the average of the CPM iterates, in order to achieve convergence. For that reason, CMWU, and our generalization, is not strictly covered by Nemirovski's results.
The key insight is to realize that if we are only interested in the regret of the players, then no averaging is needed, in which case one can show that the CPM method achieves constant regret, even in the case of a non-monotone operator, as is the case for the operator associated to general-sum games.

Using our perspective on CPM as a method for computing a sequence of strategies with low regret, we show that it is possible to generalize the CMWU method, which only applies to normal-form games, to games with arbitrary convex and compact decision sets, and utility functions that are concave with bounded gradients that are also Lipschitz continuous.
This answers an open problem of \citet{Piliouras21:Optimal}, where they ask whether CMWU could be generalized to exactly such a setting.
We call the resulting algorithm \emph{Clairvoyant Online Mirror Descent} (COMD).
We stress that while this particular result for the algorithm is new, the algorithm is really an instantiation of the CPM of \citet{Nemirovski04:Prox}.
By appealing to the contraction argument of \citet{Nemirovski04:Prox}, we also show that it is possible to improve the regret bound of CMWU in the case of normal-form games: we improve the dependence on the number of players $n$ from being linear to only $\sqrt{n}$.

Finally, we go on to develop concrete bounds for the COMD method in the case where only a finite number of steps of the contraction are performed. We show that in this case one can perform $O(\log t)$ steps of the contraction at iteration $t$, while retaining the guarantee of constant regret.
An immediate consequence of our result is that by instantiating this approximate variant of COMD with a dilated entropy regularizer~\citep{Hoda10:Smoothing,Kroer20:Faster,Farina21:Better}, we get the first algorithm that has an $O(\log T / T)$ rate of convergence to a normal-form coarse correlated equilibrium for extensive-form games, while requiring only one gradient computation and a linear-time update at every iteration.

We summarize the major strengths and weaknesses of CMWU and COMD compared to prior learning-based algorithms to compute CCE in normal-form (and, in some cases, general convex games) in \cref{tab:comparison}.

\section{Setting}

In this section we recall the definition of our setting, convex games, as well as some of the basic properties of proximal operators.

\subsection{Convex Games and Variational Inequality}

We let $\range{n} \defeq \{1, 2, \dots, n\}$ be a set of players, with $n \in \N \defeq \{1, 2, \dots\}$. In this note, we operate on \emph{convex games}, whereby each player $i \in \range{n}$ has a nonempty convex and compact set of strategies $\cX_i \subseteq \R^{d_i}$. For a \emph{joint strategy profile} $\vec{x} = (\vec{x}_1, \dots, \vec{x}_n) \in \bigtimes_{j=1}^n \cX_j$, the reward of player $i$ is given by a differentiable concave utility function $u_i : \bigtimes_{j=1}^n \cX_j \to \R$, subject to the following standard assumptions:
\begin{enumerate}%
          \item (Concavity) $u_i(\vec{x}_i, \vec{x}_{-i})$ is \emph{concave} in $\vec{x}_i$ for any $\vec{x}_{-i} = (\vec{x}_1, \dots, \vec{x}_{i-1}, \vec{x}_{i+1}, \dots, \vec{x}_n) \in \bigtimes_{j \neq i} \cX_j$;
          \item (Bounded gradients) for any $(\vec x_1, \dots, \vec x_n) \in \bigtimes_{j=1}^n \cX_j$, $\nabla_{\vec x_i} u_i(\vec x_1, \dots, \vec x_n)$ is bounded,
          \item ($\lip$-smoothness) The gradient $\nabla_{\vx_i} u_i$ is Lipschitz-continuous.
\end{enumerate}

In the rest of the note, we will often find it beneficial to view strategy updates in the game as \emph{global}, that is, operating on all players at the same time rather than each player individually. For that reason, we now introduce notation to operate on the Cartesian product of all strategy spaces. First, we denote the sum of dimensions of the strategy spaces of the players with the letter $d \defeq d_1 + \dots + d_n$. The joint strategy space of the game is $\cZ \defeq \cX_1 \times \dots \times \cX_n$, and we will consistently denote elements in $\cZ$ using the letter $\vz$ or variants thereof. Given a vector $\vz \in \bbR^d$, we will denote as $\vz_i \in \bbR^{d_i}$ the portion of the vector belonging to player $i$, that is, we let $(\vz_1, \dots, \vz_n) \in \bbR^{d_1} \times \dots \times \bbR^{d_n}$ be the (unique) vectors such that $\vz = (\vz_1, \dots, \vz_n)$.

In this global notation over the game, a key quantity associated with the game is the vector-valued function mapping strategies to payoff gradients for all players, that is,
\[
    F: \cZ \to \bbR^d,\qquad F(\vz) \defeq \begin{pmatrix} -\nabla_{\vx_1} u_1(\vz) \\ \vdots \\ -\nabla_{\vx_n} u_n(\vz)\end{pmatrix}.
\]
When viewed through the global lenses of the function $F$, the properties listed above imply the following:
\begin{assumption}\label{assumption:smooth}
    For an appropriate primal-dual norm pair $(\|\cdot\|, \|\cdot\|_*)$ over $\bbR^d$, the game operator $F : \cZ \to \bbR^d$ satisfies:
    \begin{itemize}
        \item (Boundedness) $\|F(\vz)\|_* \le B$ for all $\vz \in \cZ$.
        \item (Lipschitz continuity) There exists $\lip > 0$ such that $\|F(\vz) - F(\vz')\|_{*} \le \lip \|\vz - \vz'\|$.
    \end{itemize}
\end{assumption}

The variational inequality associated with the operator $F$, that is, the problem of finding $\vz \in \cZ$ such that
\[
    \langle F(\vz), \vz' - \vz\rangle \ge 0 \qquad\forall \vz' \in \cZ
    \numberthis{eq:vi F}
\]
is exactly equivalent to the problem of computing a Nash equilibrium of the game (see, \emph{e.g.}, \citet[Proposition 1.4.2]{Facchinei03:Finite}). We remark that generally $F$ is  \emph{not} a monotone operator, that is, there might exist $\vz,\vz'\in\cZ$ such that
$
    \langle F(\vz) - F(\vz'), \vz - \vz'\rangle \not\ge 0.
$
Nonetheless, in this note we will be concerned with applying the conceptual prox method, which was designed for monotone operators, to the variational inequality~\eqref{eq:vi F}.

\subsection{Proximal Setup}
In this section we review some standard objects and notations that relate to proximal methods. For each player $i \in \range n$, we assume that a strongly convex regularizer $\reg_i : \cZ \to \bbR$ has been chosen. Each regularizer $\reg_i$ induces a generalized notion of distance---called \emph{Bregman divergence}---over $\cX_i$, defined as
\[
    \div[i]{\cdot}{\cdot} : \cX_i \times \cX_i \to \bbR_{\ge 0}, \qquad \div[i]{\vx}{\vx'} \defeq \reg_i(\vx) - \reg_i(\vx') - \langle\nabla \reg_i(\vx'), \vx - \vx'\rangle.
\]

We combine the regularizer $\reg_i$ for each player's strategy space into a \emph{global}, composite regularizer
\[
    \reg : \cZ \to \bbR, \qquad \reg : \vz \mapsto \reg_i(\vz_i) + \dots + \reg_n(\vz_n).
\]
Correspondingly, the Bregman divergence induced by $\reg$ is the function
\[
    \div{\cdot}{\cdot} &: \cZ \times \cZ \to \bbR_{\ge 0},\\ \div{\vz}{\vz'} &\defeq \reg(\vx) - \reg(\vz') - \langle \nabla \reg(\vz'), \vz - \vz'\rangle = \div[1]{\vz_1}{\vz'_1} + \dots + \div[n]{\vz_n}{\vz'_n}.
\]
It is easy to show that as long as each $\reg_i$ is strongly convex, then so is $\reg$. Specifically, in the rest of the note we operate under the following assumption.
\begin{assumption}\label{assumption:sc}
    The per-player regularizers $\reg_i$ are chosen so that the global regularizer $\reg : \cZ \to \bbR_{\ge 0}$ is $1$-strongly convex with respect to the norm $\|\cdot\| : \cZ \to \bbR$ for which \cref{assumption:smooth} holds.
\end{assumption}

With that, we are able to define the \emph{prox operator}, which we define for the global space $\cZ$ and global regularizer $\reg$. Given a \emph{center} $\vz \in \cZ$ and a gradient $\vg \in \bbR^\dim$, the prox operator $\prox{\vz}{\vg}$ generalizes the notion of gradient step away from $\vz$ in the direction of $-\vg$, and is defined as follows.

\begin{definition}[Prox operator]
    The prox operator associated with $\reg$ is defined as
    \[
        \prox{\vz}{\vg} = \argmin_{\vu \in \cZ} \mleft\{ \vg^\top \vu + \div{\vu}{\vz} \mright\} = \argmin_{\vu \in \cZ} \mleft\{ \langle \vg - \nabla \reg(\vz), \vu\rangle + \reg(\vu)\mright\}
        \numberthis{eq:prox def}
    \]
    for any center $\vz \in \cZ$ and a vector $\vg \in \bbR^\dim$.
\end{definition}

The following properties of proximal operators are standard in the literature (see, \emph{e.g.}, \citet{Nemirovski04:Prox}). 

\begin{restatable}[The prox operator is Lipschitz continuous]{lemma}{lemproxlipschitz}\label{lem:prox is lipschitz}
    Given any center $\vz \in \cZ$, the prox operator is Lipschitz continuous with constant $1$:
    \[
        \|\prox{\vz}{\vg} - \prox{\vz}{\vg'}\| \le \|\vg - \vg'\|_*     \qquad\quad\forall\, \vg, \vg' \in \bbR^\dim.
    \]
    We remark that the above inequality uses the primal-dual norm pair for which $\reg$ is $1$-strongly convex, which is the same primal-dual norm pair for which \cref{assumption:smooth} holds, as per \cref{assumption:sc}.
\end{restatable}
\begin{lemma}\label{lem:vi for prox}
    $\prox{\vz}{\vg} = \vz^*$ if and only if $\langle \vg - \nabla \reg(\vz) + \nabla \reg(\vz^*), \vu - \vz^*\rangle \ge 0$ for all $\vu \in \cZ$. 
\end{lemma}
\cref{lem:vi for prox} simply states the (necessary and sufficient) first-order optimality condition for the prox operator problem defined in \eqref{eq:prox def}. It immediately implies the following.

\begin{corollary}\label{cor:vi for prox}
    For all $\vz \in \cZ$ and $\vg$, one has
    \[
        \div{\vu}{\prox{\vz}{\vg}} - \div{\vu}{\vz} + \div{\prox{\vz}{\vg}}{\vz} \le \langle \vg, \vu - \prox{\vz}{\vg}\rangle \qquad\forall \vu \in \cZ.
    \]
\end{corollary}
\begin{proof}
    Expanding the definition of the Bregman divergence, the statement is equivalent to
    \[
        \langle \nabla \reg(\vz) - \nabla \reg(\prox{\vz}{\vg}), \vu - \prox{\vz}{\vg}\rangle \le \langle \vg, \vu - \prox{\vz}{\vg}\rangle,
    \]
    which in turn is equivalent to
    \[
        \langle \vg - \nabla \reg(\vz) + \nabla \reg(\prox{\vz}{\vg}), \vu - \prox{\vz}{\vg}\rangle \ge 0.
    \]
    Applying \cref{lem:vi for prox} yields the statement.
\end{proof}

\subsection{Example: Normal-Form Games}\label{sec:nf}
A notable class of convex games are \emph{normal-form} games. In normal-form games, each player $i\in\range n$ has a finite set of actions of cardinality $d_i$. For any possible combination of actions of the players, each player receives a utility, which we assume to be in the range $[-V, V]$. Players are free to select any distribution over their actions as their strategy, that is, a point in the probability simplex
\[
    \cX_i = \Delta^{d_i} \defeq \{\vx \in \bbR^{d_i}_{\ge 0}: \vec{1}^\top \vx = 1\}.
\]
Each player's utility as a function of the choice of distributions is the expected utility corresponding to actions sampled for those distributions.

We use the following primal-dual norm pair $(\|\cdot\|_\Delta, \|\cdot\|_{*\Delta})$ on the Cartesian product space $\bbR^d$
\[
    \|\cdot\|_\Delta : \bbR^d \to \bbR_{\ge 0}, &\qquad \vz \mapsto \sqrt{\|\vz_1\|_1^2 + \dots + \|\vz_n\|_1^2},\\
    \|\cdot\|_{*\Delta} : \bbR^d \to \bbR_{\ge 0}, &\qquad \vz \mapsto \sqrt{\|\vz_1\|_\infty^2 + \dots + \|\vz_n\|_\infty^2}.
\]
For this choice of norms, it is immediate to check that the operator $F$ satisfies \cref{assumption:smooth} for the choice $B \defeq \sqrt{n}\, V$ and $L \defeq \sqrt{n}\, V$. 

A standard choice of regularization for normal-form games is \emph{negative entropy}, that is, the regularizer
\[
    \reg_i : \Delta^{d_i} \ni \vx \mapsto \sum_{j = 1}^{d_i} \vx[j] \log \vx[j]
\]
for all player $i \in \range n$. Negative entropy is $1$-strongly convex with respect to the $\ell_1$ norm, and therefore the composite regularizer $\reg = \reg_1 + \dots + \reg_n$ is $1$-strongly convex with respect to the norm $\|\cdot\|_\Delta$ introduced above, in compliance with \cref{assumption:sc}. Furthermore, we remark that for any player $i$, the uniform strategy $\vc_i \defeq (1/d_i, \dots, 1/d_i)$ satisfies 
\[
    \div[i]{\vx}{\vc_i} \le \log d_i\qquad\forall\vx\in\cX_i=\Delta^{d_i}.
    \numberthis{eq:range Di nf}
\]

\section{Conceptual Prox Method}

The key observation underpinning the \emph{conceptual prox method} is the following straightforward lemma.

\begin{lemma}\label{lem:central}
    Let $t \in \N$ and $\vz\^{t-1}\in \cZ$ be arbitrary. If the point $\vz\^t$ satisfies the fixed point equation: 
    \[
        \vz\^t = \prox{\vz\^{t-1}}{\eta F(\vz\^t)},
    \]
    then
    \[
        \eta\,\langle F(\vz\^t), \vu - \vz\^t \rangle \ge \div{\vu}{\vz\^t} - \div{\vu}{\vz\^{t-1}} + \div{\vz\^t}{\vz\^{t-1}} \qquad \forall \vu \in \cZ,
        \numberthis{eq:central global}
    \]
    and in particular, for all players $i$,
    \[
        \eta\,\langle \nabla_{\vx_i} u_i(\vz\^t), \vu_i - \vz\^t_i\rangle \le -\div[i]{\vu_i}{\vz\^t_i} + \div[i]{\vu_i}{\vz\^{t-1}_i} - \div[i]{\vz\^t_i}{\vz\^{t-1}_i} \qquad\forall \vu_i \in \cX_i.
        \numberthis{eq:central per player}
    \]
\end{lemma}
\begin{proof}
    Take any $\vu \in \cZ$ and use \cref{cor:vi for prox} with $\vg = \eta F(\vz\^{t})$ and $\vz = \vz\^{t-1}$:
    \[
        \div{\vu}{\vz\^t} - \div{\vu}{\vz\^{t-1}} + \div{\vz\^t}{\vz\^{t-1}} &\le \eta\,\langle F(\vz\^{t}), \vu - \vz\^t\rangle,
    \]
    which is exactly~\eqref{eq:central global}. Since the inequality holds for any $\vu \in \cZ$, it holds in particular for any vector of the form $\vu = (\vz_1\^{t}, \dots, \vz_{i-1}\^t, \vu_i, \vz\^t_{i+1}, \dots, \vz\^t_n) \in \cZ$. Substituting this particular choice into~\eqref{eq:central global} and expanding the definitions yields~\eqref{eq:central per player}.
\end{proof}

By noting that the right-hand side of \eqref{eq:central per player} is telescopic, \cref{lem:central} immediately implies the following.

\begin{corollary}[Constant per-player regret]\label{cor:central}
    Let $\vz\^0 \in \cZ$ be arbitrary, and suppose that recursively $\vz\^t \in \cZ$ at all times $t \in \N$ satisfies the following fixed-point equation
    \[
        \vz\^t = \prox{\vz\^{t-1}}{\eta F(\vz\^t)}. \numberthis[$\star$]{eq:fp}
    \]
    Then, at all times $T \in \N$, the per-player regret is upper bounded as
    \[
        \mathrm{Reg}^T_i &\defeq \max_{\vu_i \in \cX_i} \sum_{t=1}^T \langle \nabla_{\vx_i} u_i(\vz\^t), \vu_i - \vz\^t_i\rangle \\
            &\le \frac{1}{\eta}\max_{\vu_i \in \cX_i} \mleft\{ \div[i]{\vu_i}{\vz\^0_i} - \div[i]{\vu_i}{\vz\^T_i}\mright\} - \frac{1}{\eta}\sum_{t=1}^T \div[i]{\vz\^t_i}{\vz\^{t-1}_i}\\
            &\le \frac{1}{\eta}\max_{\vu_i \in \cX_i} \mleft\{ \div[i]{\vu_i}{\vz\^0_i} \mright\}.
    \]
\end{corollary}
\begin{proof}
    The statement follows immediately from summing~\eqref{eq:central per player} for $t = 1,\dots, T$, and noticing that the terms $-\div[i]{\vu_i}{\vz\^t_i} + \div[i]{\vu_i}{\vz\^{t-1}_i}$ telescope:
    \[
        \eta \sum_{t=1}^T \langle \nabla_{\vx_i} u_i(\vz\^t), \vu_i - \vz\^t_i\rangle 
        &\overset{\eqref{eq:central per player}}{\le} \sum_{t=1}^T \mleft( -\div[i]{\vu_i}{\vz\^t_i} + \div[i]{\vu_i}{\vz\^{t-1}_i} - \div[i]{\vz\^t_i}{\vz\^{t-1}_i} \mright)\\
        &= \div[i]{\vu_i}{\vz\^0_i} - \div[i]{\vu_i}{\vz\^T_i} - \sum_{t=1}^T \div[i]{\vz\^t_i}{\vz\^{t-1}_i}
    \]
    for all $\vu_i \in \cX_i$. Dividing by $\eta$ and taking a maximum over $\vu_i\in\cX_i$ yields first inequality in the statement. The second inequality follows immediately by using the fact that divergences are always nonnegative.
\end{proof}

The algorithm defined in \cref{cor:central} is called the \emph{conceptual prox method (CPM)} (see also \citet{Nemirovski04:Prox}).
\cref{cor:central} shows that the per-player regret cumulated up to any time $T$ by the fixed-point iterates $\vz\^t = \prox{\vz\^{t-1}}{F(\vz\^t)}$ produced by the CPM is bounded by the range of the divergence $\div[i]{\cdot}{\vz\^0_i}$, a quantity independent of time. For example, when $\cX_i = \Delta^m$ is the $m$-simplex, $\vz\^0_i$ is the uniform strategy, and $\reg_i$ is negative entropy, then $\max_{\vx\in\cX_i} \div[i]{\vx}{\vz\^0_i} = \log m$.

\paragraph{Existence and Computation of Fixed-Point Solutions} At this stage, it is perhaps unclear why the fixed points~\eqref{eq:fp} exist and how one can compute them. The key lies in the following observation, which dates back to at least the work of \citet{Nemirovski04:Prox}:

\begin{observation}\label{obs:contraction}
At all times $t$ the map
\[
    \vw \mapsto \prox{\vz\^{t-1}}{\eta F(\vw)}
\]
is $\eta L$-Lipschitz continuous. Hence, as long as $\eta < 1/L$, the above function is a contraction, and the fixed point is therefore unique. Consequently, convergence to an $\epsilon$-fixed point can be achieved via a number of iterations that scales proportially to $\log (1/\epsilon)$.
\end{observation} 

This is straightforward: the proximal operator itself is $1$-Lipschitz continuous (\cref{lem:prox is lipschitz}), and $\eta F$ is $\eta L$-Lipschitz continuous, so their composition is $\eta L$-Lipschitz continuous. 
So, at least \textit{approximate} fixed-point solutions are easy to compute, showing that at least in an approximate sense the conceptual prox method can be computed. In the next section we quantify the error introduced by the approximation in the fixed-point solution.

\section{Conceptual Prox Method with Approximate Fixed Points}
A more refined analysis of the argument employed in \cref{cor:central} takes into account error in the computation of the fixed-point solutions~\eqref{eq:fp}. We start by relaxing \cref{lem:central}.

\begin{lemma}\label{lem:central apx}
    Let $t \in \N$ and $\vz\^{t-1}\in \cZ$ be arbitrary. Let $\vw\^t \in \cZ$ be an approximate fixed point, in the sense that for some $\epsilon\^t \ge 0$,
    \[
        \mleft\| \vw\^t - \prox{\vz\^{t-1}}{\eta F(\vw\^t)} \mright\| \le \epsilon\^t.
        \numberthis{eq:apx}
    \]
    Then, the point $\vz\^t \defeq \prox{\vz\^{t-1}}{\eta F(\vw\^t)}$ satisfies
    \[
        \eta\,\langle F(\vw\^t), \vu - \vw\^t \rangle \ge \div{\vu}{\vz\^t} - \div{\vu}{\vz\^{t-1}} + \div{\vz\^t}{\vz\^{t-1}} - \eta B \epsilon\^t \qquad \forall \vu \in \cZ,
        \numberthis{eq:central global apx}
    \]
    and in particular, for all players $i$,
    \[
        \eta\,\langle \nabla_{\vx_i} u_i(\vw\^t), \vu_i - \vw\^t_i\rangle \le -\div[i]{\vu_i}{\vz\^t_i} + \div[i]{\vu_i}{\vz\^{t-1}_i} - \div[i]{\vz\^t_i}{\vz\^{t-1}_i} + \eta B \epsilon\^t \qquad\forall \vu_i \in \cX_i.
        \numberthis{eq:central per player apx}
    \]
\end{lemma}
\begin{proof}
    The proof of the second part of the statement is identical to that of \cref{lem:central}. Hence, we focus on proving~\eqref{eq:central global apx}. Again, we start from \cref{cor:vi for prox}, this time applied with 
    $\vg = \eta F(\vw\^{t})$ and $\vz = \vz\^{t-1}$:
    \[
        \div{\vu}{\vz\^t} - \div{\vu}{\vz\^{t-1}} + \div{\vz\^t}{\vz\^{t-1}} &\le \eta\,\langle F(\vw\^{t}), \vu - \vz\^t\rangle\\
        &= \eta\,\langle F(\vw\^{t}), \vu - \vw\^t\rangle + \eta\,\langle F(\vw\^{t}), \vw\^t - \vz\^t \rangle\\
        &\le \eta\,\langle F(\vw\^{t}), \vu - \vw\^t\rangle + \eta \|F(\vw\^{t})\|_* \| \vw\^t - \vz\^t \| \\
        &\le \eta\,\langle F(\vw\^{t}), \vu - \vw\^t\rangle + \eta B \epsilon\^t,
    \]
    where the second inequality follows from the definition of the dual norm, and the last inequality from using the definition of $B$, introduced in \cref{assumption:smooth}, and~\eqref{eq:apx}. Rearranging yields~\eqref{eq:central global apx}.
\end{proof}

Repeating the analysis we already carried out for~\cref{cor:central}, this time using~\cref{lem:central apx}, we obtain the following.

\begin{corollary}\label{cor:central with error}
    Let $\vz\^0 \in \cZ$ be arbitrary, and recursively let $\vz\^t, \vw\^t \in \cZ$ at all times $t \in \N$ be such that
    \[
        \mleft\| \vw\^t - \prox{\vz\^{t-1}}{\eta F(\vw\^t)} \mright\| \le \epsilon\^t, \qquad 
        \vz\^t \defeq \prox{\vz\^{t-1}}{\eta F(\vw\^t)}.\numberthis[$\star\star$]{eq:fp apx}
    \]
    Then, at all times $T \in \N$, the per-player regret associated with iterates $\vw\^t$ is upper bounded as
    \[
        \mathrm{Reg}^T_i &\defeq \max_{\vu_i \in \cX_i} \sum_{t=1}^T \langle \nabla_{\vx_i} u_i(\vw\^t), \vu_i - \vw\^t_i\rangle \\
            &\le \frac{1}{\eta}\max_{\vu_i \in \cX_i} \mleft\{ \div[i]{\vu_i}{\vz\^0_i} - \div[i]{\vu_i}{\vz\^T_i}\mright\} - \frac{1}{\eta}\sum_{t=1}^T \div[i]{\vz\^t_i}{\vz\^{t-1}_i} + B \sum_{t=1}^T \epsilon\^t\\
            &\le \frac{1}{\eta}\max_{\vu_i \in \cX_i} \mleft\{ \div[i]{\vu_i}{\vz\^0_i} \mright\} + B\sum_{t=1}^T \epsilon\^t.
    \]
\end{corollary}

\begin{observation}\label{obs:epsilon}
When the choice $\epsilon\^t = 1/t^2$ is used, then the sum of errors $\sum_{t=1}^T \epsilon\^t$ is an additive constant bounded by $2$, and therefore does not affect the constant per-player regret guarantees, while at the same time requiring $O(\log t)$ fixed-point iterations per iteration of the learning algorithm.    
\end{observation}

\subsection{Centralized Implementation}

By combining \cref{cor:central with error} together with the concrete choice of errors $\epsilon\^t$ given in \cref{obs:epsilon} and fixed-point iterations on the map $\Pi_{\vz\^{t-1}} \circ \eta F$, which is a contraction for $\eta \le 1/(2L)$ (\cref{obs:contraction}), we obtain \cref{algo:centralized}, whose properties summarized in \cref{thm:centralized} follow directly from the preceding discussion.

\begin{algorithm}[ht]
    \caption{Conceptual prox method with approximate fixed points (\textit{centralized} implementation)}
    \label{algo:centralized}
    \DontPrintSemicolon
    \KwData{$\vz\^0 \in \cZ$ initial point, $0 < \eta \le 1/(2L)$ learning rate, $\epsilon\^t$ desired fixed-point approximation error\!\!\!\!\!\!}
    \BlankLine
    \For{$t=1,2,\dots$}{
        $\vw\^t \gets \vz\^{t-1}$\;
        \While{$\displaystyle\mleft\|\vw\^t - \prox{\vz\^{t-1}}{\eta F(\vw\^t)}\mright\| > \epsilon\^t$}{\vspace{2mm}
            $\vw\^t \gets \prox{\vz\^{t-1}}{\eta F(\vw\^t)}$
            \label{line:prox}
            \Comment*{\color{commentcolor}Fixed-point iteration]\!\!\!\!}
        }
        \vspace{2mm}
        $\vz\^t \gets \prox{\vz\^{t-1}}{\eta F(\vw\^t)}$\;
    }
\end{algorithm}

\begin{theorem}\label{thm:centralized}
    At all times $t = 1,2, \dots$ in \cref{algo:centralized},
    \begin{enumerate}
        \item The internal {\normalfont\textbf{while}} loop runs for at most $\log_2(\max_{\vz,\vz'\in\cZ}\|\vz - \vz'\|) + \log_2 \frac{1}{\epsilon\^t}$ iterations.
        \item The iterates $\vz\^t$ produced by the algorithm achieve regret upper bounded by
        \[
            \mathrm{Reg}_i^t \defeq \max_{\vu_i \in \cX_i} \sum_{\tau=1}^t \langle \nabla_{\vx_i} u_i(\vw\^\tau), \vu_i - \vw\^\tau_i\rangle\le \frac{1}{\eta} \max_{\vu_i \in \cX_i} \mleft\{ \div[i]{\vu_i}{\vz\^0} \mright\} + B \sum_{\tau=1}^t \epsilon\^\tau
        \]
        for each player $i$. Correspondingly, the average product distribution of play $\frac{1}{t} \sum_{\tau = 1}^t \vz_1\^\tau \otimes \dots \otimes \vz_n\^\tau$
        is a $\kappa$-coarse correlated equilibrium of the game, with
        \[
            \kappa = \frac{\mathrm{Reg}_i^t}{t} \le \frac{1}{\eta\,t} \max_{\vu_i \in \cX_i} \mleft\{ \div[i]{\vu_i}{\vz\^0} \mright\} + \frac{B}{t}\sum_{\tau=1}^t \epsilon\^\tau
        \]
    \end{enumerate}
\end{theorem}
Once again, we remark that the choice $\epsilon\^t = 1/t^2$ for all $t$ is natural, and leads to constant per-player regret, as well as convergence to a coarse correlated equilibrium of the convex game at the rate $1/t$. Alternatively, if the total number of iterations $T$ was known in advance, the choice $\epsilon\^t = 1/T$ would lead to a similar result.

We refer to \cref{algo:centralized} as a \emph{centralized} implementation because it operates directly on the product space $\cZ$. We believe that this is the natural setting in which to analyze the algorithm. However, as we lay out in the next section, the individual steps that make \cref{algo:centralized} can be expanded to have a learning dynamic flavor for each player, with some important caveats.

\subsection{Decentralized Implementation: Clairvoyant OMD}

In this section we show that \cref{algo:centralized} can be implemented in the form of a decentralized learning algorithm where each player independently updates their strategy upon observing the gradient of their utility. The key is in the observation that, given the definition of $\reg = \reg_1 + \dots + \reg_n$, the proximal operator on \cref{line:prox} decomposes as
\[
\prox{\vz\^{t-1}}{\eta F(\vw\^t)} = \begin{pmatrix}
    \displaystyle\argmin_{\vx_1\in\cX_1} \mleft\{ -\eta \langle \nabla_{\vx_1} u_1(\vw\^t), \vx_1\rangle + \div[1]{\vx_1}{\vz\^{t-1}_1}\mright\}\\
    \vdots \\
    \displaystyle\argmin_{\vx_n\in\cX_n} \mleft\{-\eta \langle \nabla_{\vx_n} u_n(\vw\^t), \vx_n\rangle + \div[n]{\vx_n}{\vz\^{t-1}_n}\mright\}\\
\end{pmatrix},
\]
which corresponds to an OMD update (with linearized losses) for each player. Then, by fixing the number of fixed-point iterations (that is, repetitions of the \textbf{while} loop) to the quantity
\[
    N\^t \defeq 1 + \log_2\mleft(\max_{\vz, \vz'\in \cZ} \|\vz - \vz'\|\mright) + \log_2 \frac{1}{\epsilon\^t},
\]
it becomes guaranteed that the approximation error of the approximate fixed-point generated after $N\^t - 1$ iterations is less than $\epsilon\^t$, and thus we can rewrite \cref{algo:centralized} as in \cref{algo:decentralized}.

\vspace{2mm}
\begin{algorithm}[H]
    \caption{Decentralized, per-player implementation (Clairvoyant OMD)}
    \label{algo:decentralized}
    \DontPrintSemicolon
    \KwData{$\vz_i\^0 \in \cX_i$ initial strategy for each player, $0 < \eta \le 1/(2L)$ learning rate, $\epsilon\^t$ desired fixed-point approximation error\!\!\!\!\!\!}
    \BlankLine
    $t' \gets 0$\;
    $\vw\^{t'} \gets \vz\^{0}$\;
    \For{$t=1,2,\dots$}{
            \Comment{\color{commentcolor} Begin unrolling of \textbf{while} loop on \cref{line:prox} of \cref{algo:centralized}]}
            \For{$k=1,\dots, N\^t$}{\vspace{2mm}
                \For{each player $i\in\range n$, in parallel}{\vspace{2mm}
                    $t' \gets t' + 1$\;
                    $\vw\^{t'}_i \gets \displaystyle\argmin_{\vx_i\in\cX_i} \mleft\{ -\eta \langle \nabla_{\vx_i} u_i(\vw\^{t'-1}_1, \dots, \vw\^{t'-1}_n), \vx_i\rangle + \div[i]{\vx_i}{\vz\^{t-1}_i}\mright\}$\label{line:omd}\;
                }
            }
            \Comment{\color{commentcolor}End unrolling of \textbf{while} loop on \cref{line:prox} of \cref{algo:centralized}]}
            \vspace{1mm}
            $\vz\^t \gets \vw\^{t'}$\;
        }
\end{algorithm}
\vspace{2mm}

As \cref{algo:decentralized} is amounts to an alternative implementation of \cref{thm:centralized}, the guarantees of \cref{thm:centralized} apply to \cref{algo:decentralized} as well. We also remark that in the latter implementation, each player $i\in\range n$ crucially updates their approximate fixed-point strategy $\vw_i\^{t'}$ using as gradient vector the gradient of their utility evaluated in the strategy profile $(\vw\^{t'-1}_1,\dots,\vw\^{t'-1}_n)$. This is perfectly consistent with the framework of learning in games, where each player updates their strategy based on their gradient at the current strategy profile.

One caveat with this implementation is that only the iterates $\vz\^t$ (and not the approximate fixed-point iterates $\vw\^{t'}$) are guaranteed to cumulate low regret. In other words, in the learning-in-games interpretation, only the subsequence of iterates $\{\vw\^{N\^1}, \vw\^{N\^1 + N\^2}, \dots\}$ is guaranteed to have low regret. Therefore, when interpreted as a learning algorithm, \cref{algo:decentralized} provides uncoupled learning dynamics, but is not a \emph{no-regret} algorithm in the classic sense. Nonetheless, a predictable subsequence of iterates guarantees constant regret, and therefore the algorithm can be used to extract an approximate coarse correlated equilibrium. 

The name \emph{Clairvoyant OMD} was chosen to explicitly point out that the algorithm is a generalization, to general convex games, of the Clairvoyant MWU algorithm recently introduced by \citet{Piliouras21:Optimal} in the special case of normal-form games and for certain specific choices of regularizers---see also \cref{sec:comd simplex}. We remark also that while Clairvoyant MWU was originally introduced as a \emph{centralized} algorithm, the latter decentralized interpretation was preferred in the later revision of the paper by the same authors \citep{Piliouras22:Beyond}. 

\section{The Special Case of Clairvoyant MWU}\label{sec:comd simplex}

In the special case of normal-form games where each player's simplex strategy space $\Delta^{d_i}$ has been equipped with the negative entropy regularizer (\cref{sec:nf}), the OMD-like update step on \cref{line:omd} of \cref{algo:decentralized} has the closed-form solution
\[
    \vw\^k_i[a] \propto \vz_i\^{t-1}[a]\cdot \exp\mleft\{\eta\cdot \nabla_{\vx_i} u_i(\vw_1\^{k-1}, \dots, \vw_n\^{k-1})[a]\mright\}
\]
By plugging the above closed formula into \cref{algo:centralized} and \cref{algo:decentralized}, we recover the centralized and decentralized versions of the Clairvoyant MWU algorithm introduced by \citet{Piliouras21:Optimal,Piliouras22:Beyond}.

As already mentioned in \cref{sec:nf}, in that setting the game operator $F$ is upper bounded (with respect to the dual norm $\|\cdot\|_{*\Delta}$) by $B$ and is $\lip$-Lipschitz continuous with respect to the $(\|\cdot\|_\Delta, \|\cdot\|_{*\Delta})$ norm pair, where $B = L = \sqrt{n}\,V$ and $V$ is the maximum absolute utility for any player in the game. Furthermore, negative entropy is $1$-strongly convex with respect to $\|\cdot\|_\Delta$. Hence, by choosing
\[
    \eta \defeq \frac{1}{2\sqrt{n}\, V},\qquad \epsilon\^t \defeq  \frac{1}{t^2},
\]
from \cref{thm:centralized} we recover a per-player regret of
\[
    \mathrm{Reg}_i^T \le 2\sqrt{n}\,V (1 + \log(d_i)) = O(\sqrt{n} V \log(d_i))
\]
for the iterates $\vz\^1, \dots, \vz\^T$ produced by \cref{algo:centralized,algo:decentralized}. %
The number of intermediate iterates $\vw\^{t'}$ produced by \cref{algo:decentralized} in this case is
\[
    \sum_{t=1}^T N\^t &= \sum_{t=1}^T \mleft( 1 + \log_2 \frac{1}{\epsilon\^t} + \log_2 \max_{\vz,\vz' \in \cZ} \|\vz - \vz'\|_\Delta\mright) = O(T \log T + T \log n),
\]
where the last step uses the fact that
\[
    \max_{\vz,\vz' \in \cZ} \|\vz - \vz'\|_\Delta = \sqrt{\sum_{i=1}^n \max_{\vx,\vx' \in \Delta^{d_i}}\|\vx - \vx'\|_1^2} \le 2\sqrt{n}.
\]
This shows that from $O(T \log T + T \log n)$ iterates of $\vw\^{t'}$ it is possible to extract a subsequence of $T$ iterates (those corresponding to $\vw\^{N\^1}, \vw\^{N\^1 + N\^2}, $ and so on), that cumulate $O(\sqrt{n} V \log(d_i))$ regret.
This regret bound refines that of \citet{Piliouras22:Beyond} in the dependence on the number of players ($\sqrt{n}$ rather than $n$).

\section*{Acknowledgments}
Haipeng Luo is partially supported by the NSF under award IIS-1943607.
Christian Kroer is supported by the Office of Naval Research Young Investigator Program under grant N00014-22-1-2530.

\bibliography{main}
\end{document}